\documentclass[journal,10pt,twocolumn]{IEEEtran}
\usepackage{epsfig}
\usepackage{graphicx}
\usepackage{subfigure}

\newtheorem{proposition}{Proposition}
\newtheorem{lemma}{Lemma}

\newtheorem{algorithm}{Algorithm}

\newtheorem{corollary}{Corollary}

\usepackage{amssymb}
\usepackage{makeidx}
\usepackage{amsmath}
\usepackage{graphicx}
\usepackage{epsf}
\usepackage{epsfig}
\usepackage{psfig}
\usepackage{ccaption}
\usepackage{array}
\usepackage{tabularx}
\usepackage{multirow}
\usepackage{epsfig}
\usepackage{cite}
\usepackage{procedure,algorithm,algorithmic}

\begin{document}

\title{Taming and Leveraging Interference in Mobile Radar Networks}
\author{Husheng Li}
\maketitle

\begin{abstract}
Mobile radar networks, such as autonomous driving systems, are subject to the severe challenge of mutual interference. Despite the inborn interference-proof capability in frequency modulation continuous waveform (FMCW) radar, interference management is necessary for dense radar networks. The approaches combatting the radar interference include the timing diversity of waveform parameters, data-driven detection interference, and signal processing based decoupling of interference. Moreover, the leverage of radar interference for inferring the information of interfere is studied. Simulations are carried out in the scenario of radar interferences in vehicular networks, which demonstrates the performance of the proposed algorithms. 

\end{abstract}

\section{Introduction}

Radar technology, as a powerful tool to sense the environment, has found many applications in modern cyber physical systems (CPSs), such as unmanned aerial vehicular (UAV) networks and autonomous vehicle networks. In these systems, each node uses radar to monitor the positions and speeds of neighboring nodes, as well as other targets in the environment, for the avoidance of collision or formation control. When there are multiple radar transceivers, they are subject to mutual interference (either direct or indirect), as illustrated in Fig. \ref{fig:cars}. Different from wireless data communications, in which interference is a major concern, there have been relatively few studies on the interference mitigation of radar signals. Although radar systems have a certain inherent capability of combating interference (e.g., being removed by the smoothing capability of Kalman filtering, or filtered out after being mixed with the local oscillator in frequency modulation continuous waveform (FMCW) radar receivers), in new applications in which there are unprecedentedly many radar transceivers (e.g., there could be hundreds of radars within hundreds of meters for UAV networks or autonomous driving networks), the interference could be very intensive and be beyond the traditional interference-proof capability of radar systems. As has been shown by existing studies, the interference can incur significant estimation errors of positions and speeds, which could be detrimental to the corresponding applications, e.g., collisions of vehicles or the loss of UAV formations. Therefore, there is a pressing demand for the study of interference management in radar networks, as indicated in the European Commission report \cite{MOSARIM2012}.

Despite the pressing need and various preliminary researches, there still lacks a systematic study on the mitigation of radar interference, in both the physical (PHY) and medium access (MAC) layers. For the PHY layer, in existing studies, simple approaches have been adopted, such as detecting the interference by abnormally high power and then eliminating the interference using a notch filter in the time or frequency domain. There are even less studies in the MAC layer. Most existing studies on radar interference management are focused on the PHY layer; however, in the context of massive radar networks, it is of key importance to efficiently allocate the sensing resource, namely power and bandwidth to different radar transceivers, thus requiring efficient solutions in the MAC layer. To avoid interference, different radar transceivers can be allocated different frequency bands, thus realizing an FDMA-like solution. However, modern radar systems usually use large bandwidth for accurate sensing performance; e.g., a bandwidth of 2GHz is used by the 77GHz millimeter wave (mmWave) radar. Even the large bandwidth of mmWave band cannot assure interference avoidance using orthogonal spectrum usage for many (say 100) radars. Hence, the MAC layer design of radar networks needs to take the inevitable impact of interference into account.

\begin{figure}[!t]
\centering
\includegraphics[width=2.7in]{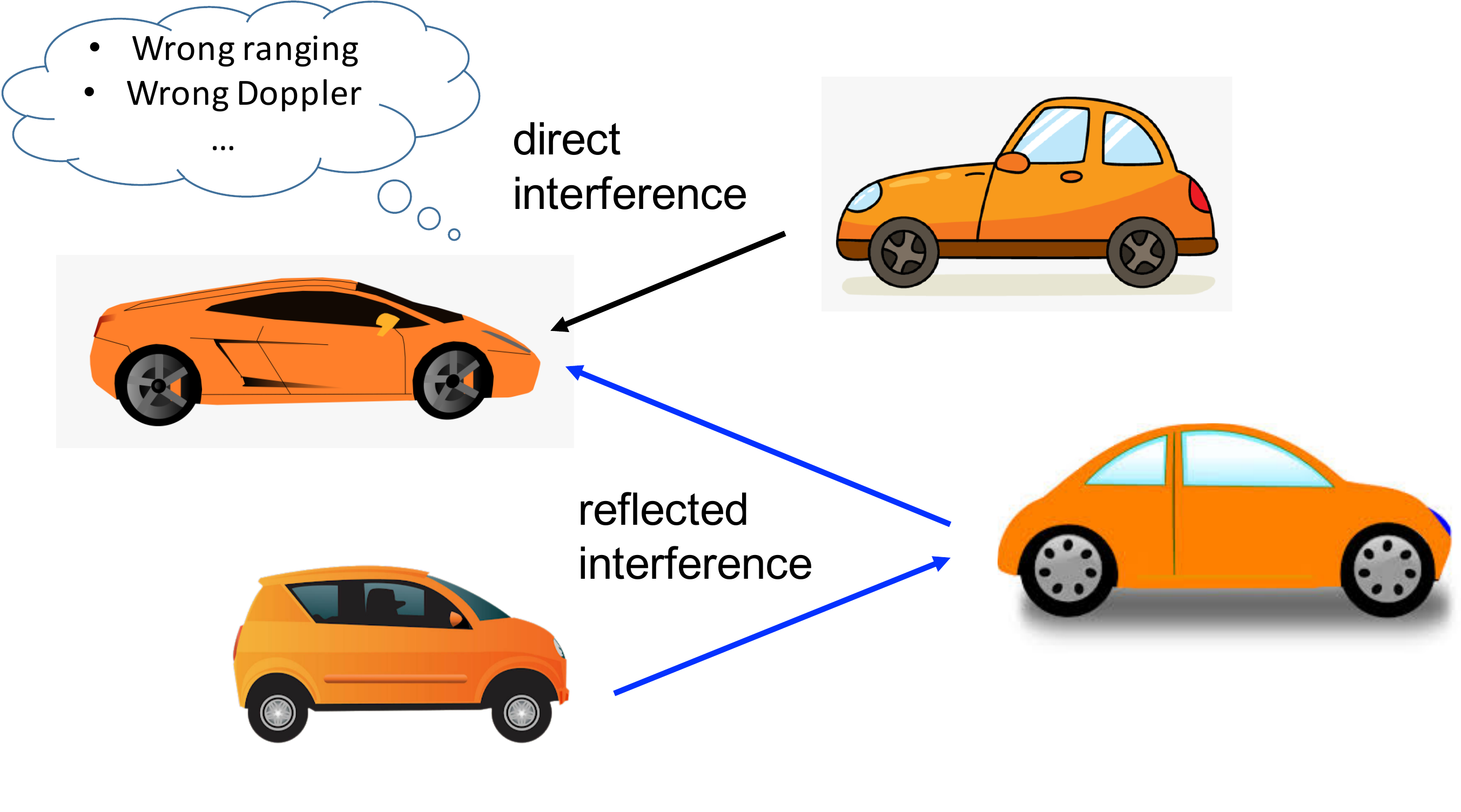}
\caption{An illustration of the interference of vehicle radars}\label{fig:cars}
\label{pattern}
\end{figure}

In this paper, we provide a systematic study on the mitigation and leverage of radar interference. The following approaches will be studied:
\begin{itemize}
\item Interference avoidance and average: Various approaches of diversifying the waveform parameters are used to either reduce or average the radar interference. 
\item Interference detection: This is to determine whether the received signal is interference, which could be based on the features of interference and be trained using data-driven approach.
\item Interference decoupling: Using the framework of multiuser sensing proposed in this paper, the signal and interference are expected to be decoupled for further analysis. 
\item Interference leverage: When the interference can be detected and decoupled, it provides information about the interferer, which converts the harm to benefit. 
\end{itemize}

The remainder of this paper is organized as follows. The related research is introduced in Section \ref{sec:related}. Then, the system model of radar network and interference are explained in Section \ref{sec:model}. The probabilistic analysis on the interference is given in Section \ref{sec:analysis}. The detection and avoidance of radar interference are discussed in Sections \ref{sec:detection} and \ref{sec:avoidance}, respectively. The methodology of leveraging radar interference for information inference is introduced in Section \ref{sec:leverage}. A framework of multiuser sensing is detailed in Section \ref{sec:multiuser}. The numerical simulation results and conclusions are provided in Sections \ref{sec:numerical} and \ref{sec:conclusions}, respectively. 

\section{Related Works}\label{sec:related}
In this section, we introduce the existing research related to this paper.

\subsection{Radar Interference} In the MOSARIM report \cite{MOSARIM2012}, seven scenarios of radar interference and nine possible solutions have been discussed. The corresponding feasibility has been analyzed with either experimental results or simulations. A more recent report on the same issue is given in \cite{NHTSA2018}, which summarized the radar interference in different scenarios. Various approaches have been proposed for different types of radar systems. In the MOSARIM report, it is claimed that the constant false alarm rate (CFAR) receiver is the most effective to radar interference by decreasing the radar sensitivity, which has been adopted in the Texas Instrument radar due to its simplicity \cite{Mani2019}. However, its degradation of radar performance has been questioned in the US-DOT report \cite{NHTSA2018}. Comprehensive surveys can be found in \cite{Aydogdu2020}. A simple approach is to detect abnormally high received power and then use a notch filter to remove the interference \cite{Brooker2007,Bechter2015}. The structure of FMCW signal is used to interpolate the original waveform for interference cancellation. When multiple antennas are available, legitimate signal can use a different subspace from the interference \cite{Tian2018}. Deep learning has been used to identify radar interference \cite{Mun2018}. In the MAC layer, either frequency hopping \cite{Luo2013} or orthogonal sequence \cite{Skaria2019} has been proposed to avoid interference. There are also some theoretical studies on the damage of interference, e.g., the stochastic geometric analysis in \cite{Munari2018,Hourani2018,Chu2020,Ram2020}. In particular, the ghost image \cite{Goppelt2011,Kim2018} is a common phenomenon due to radar interference: a false target is detected by a radar interfered by another radar transmitter, due to the incoherence of phase. The mutual interference between radar and communications has been studied in \cite{Aydogdu2019}. In a contrast to the schemes proposed by us, which includes stagewise management and information inference using interference, these existing studies focus on only one stage of interference avoidance of individual radars, and do not consider the benefit from networking and the interference itself; in addition, they are more ad hoc and do not have a systematic mathematical framework.

\subsection{Communication Interference Management} Interference is a major concern in communication networks. Various multiple access schemes, such as TDMA or OFDM, can avoid interference due to the orthogonal usage of time or frequency \cite{Proakis2018}. In multi-cell cases, frequency reuse can be adopted such that neighboring cells may use different frequency bands according to the reuse factor \cite{JunyiLi2013}. When interference is inevitable, interference averaging, e.g., using frequency hopping or CDMA, can help to diversify the interference, in order to avoid long-term dominating interference \cite{Tse2005}. Meanwhile, power control is an important technique to reduce unnecessary interference \cite{Goodman2000,Yates1995}. The cross-layer design is a powerful mechanism for interference management in wireless communication networks \cite{Kelly1998,Chiang2007,Palomar2006,XiaojunLin2006}. All these techniques, which have been widely employed in practical communication systems, provide insights and motivations for the interference management in radar networks. However, communication interferences have significantly different physical mechanisms, and can usually be well coordinated by fixed base stations (while being infeasible in mobile radar networks); therefore, the corresponding interference management schemes cannot be applied to radar networks straightforwardly.

\subsection{Radar Jamming/Deception} The interference of radar could be intentional, e.g., the jamming in a denial-of-service manner and deception. There have been substantial studies on jamming and deception of radar signals \cite{Schuerger2008,Schuerger2009,Greco2008,Maithripala2005,Bo2011}. The radar jamming \cite{Schuerger2008,Schuerger2009,Greco2008} is to use high power to suppress the legitimate radar signal, while the deception is to spoof the radar receiver and use fake signal to mislead the sensing procedure \cite{Maithripala2005,Bo2011}. Frequency hopping is an effective approach to combat the radar jamming, while authentication in the signal level is the way to avoid deception. The proposed project addresses unintentional interference, thus being different from the radar jamming and deception. 

\section{System Model}\label{sec:model}
In this section, we introduce the model of the radar network and the radar operation. 

\subsection{Motion Dynamics}
We consider $N(t)$ mobile radar transceivers in a radar network, in which radar transceivers can communicate to neighboring ones. Here $N(t)$ changes with time, namely there could be leaving or arriving transceivers. The motion law of the transceivers is not specified, except for the detailed simulations in the numerical results. For simplicity, we consider the motions of the radar transceivers in the plane $\mathbb{R}^2$. 

\subsection{Radar Waveform and Reception}
We assume that each radar transceiver sends out beams of radar waveforms. We assume FMCW radar waveforms, whose transmit waveform is given by (assuming that a pulse is sent out at time 0)
\begin{small}
\begin{eqnarray}
p(t)=\sum_{n=-\infty}^{\infty}AW(t-nT_p)\sin\left(\int_{0}^{t-nT_p} 2\pi f(s) ds+\theta_n\right),
\end{eqnarray}
\end{small} 
where $T_p$ is period of radar pulses, $W$ is the rectangle window function between time 0 and $T_c$, $T_c$ is the time duration of each radar pulse (chirp) and $\theta_n$ is the phase in the $n$-th pulse. The instantaneous frequency $f(t)$ is given by
\begin{eqnarray}
f(t)=f_0+St,
\end{eqnarray}
where $f_0$ is the initial frequency and $S$ is the increasing slope of the frequency. 

When multiple radar transceivers are considered, the time offsets of the $N$ transceivers are denoted by $\tau_n$, $n=1,...,N$, where we assume $\tau_1=0$. At transceiver $n$, the received signal is given by
\begin{eqnarray}
r_n(t)&=&\sum_{m=1}^N\sum_{n=-\infty}^{\infty}A_{mn}W(t-nT_p-\tau_m-\tau_{mn})\nonumber\\
&&\sin\left(\int_{-\infty}^{t-nT_p-\tau_m-\tau_{mn}} 2\pi f(s) ds+\theta_m\right),
\end{eqnarray} 
where $\tau_{mn}$ is the traveling time of the pulse from transceiver $m$ to $n$ (possibly reflected at an intermediate reflector) and $A_{mn}$ is the corresponding amplitude of the received signal. 

\subsection{Ranging}
When the reflected signal is intercepted by the radar transceiver, it is mixed with a local oscillation, as illustrated in Fig. \ref{fig:FMCW}, whose output is proportional to (suppose that the starting time is 0 and there is no interference)
\begin{small}
\begin{eqnarray}\label{eq:beat}
&&o(t)\propto \underbrace{\cos\left(2\pi S\tau t-\pi S\tau^2+2\pi f_0\tau\right)}_{\text{beat signal: freq. = input freq. - local freq.}}\\
&+&\underbrace{\cos\left(\pi S(t-\tau)^2+2\pi f_0(t-\tau)+\pi St^2+2\pi f_0t+2\theta_0\right)}_{\text{high frequency, removed by IF filter}},\nonumber
\end{eqnarray}
\end{small}
where $\tau$ is the round trip time of the FMCW pulse. The second term has a high frequency and will be filtered out by an IF filter. The first term is a sinusoidal function with beat frequency $S\tau$. This procedure is illustrated in Fig. \ref{fig:FMCW}. By estimating the beat frequency $S\tau$, the radar transceiver evaluates the round trip time and computes the distance by $0.5c\tau$, where $c$ is the light speed. The target velocity can be estimated from the corresponding Doppler shift. 

An interference will add perturbation to the IF filter output in (\ref{eq:beat}), thus impairing the subsequent processing. For example, when an FMCW radar is interfered by another FMCW radar having the same frequency slope $S$, the output of the IF filter is given by
\begin{small}
\begin{eqnarray}\label{eq:beat2}
o'(t)=A_1 \underbrace{\cos\left(2\pi S\tau t+\theta\right)}_{\text{legitimate signal}}+
A_2\underbrace{\cos\left(2\pi S\tau' t+\theta'\right)}_{\text{interference}},
\end{eqnarray}
\end{small}
where $A_1$ and $A_2$ are the amplitudes of signal and interference, respectively, $\theta$ and $\theta'$ are the corresponding phases, and $\tau'$ is the time offset between the interfering radar pulse and local oscillation (including the propagation time and starting time difference). When $A_2$ is significant, the estimation on the beat frequency $S\tau$ (thus $\tau$ and the distance) will be substantially impaired. This framework can also be extended to other radar waveforms, such as OFDM radar and pulse radar.

\begin{figure}[!t]
\centering
\includegraphics[width=0.45\textwidth]{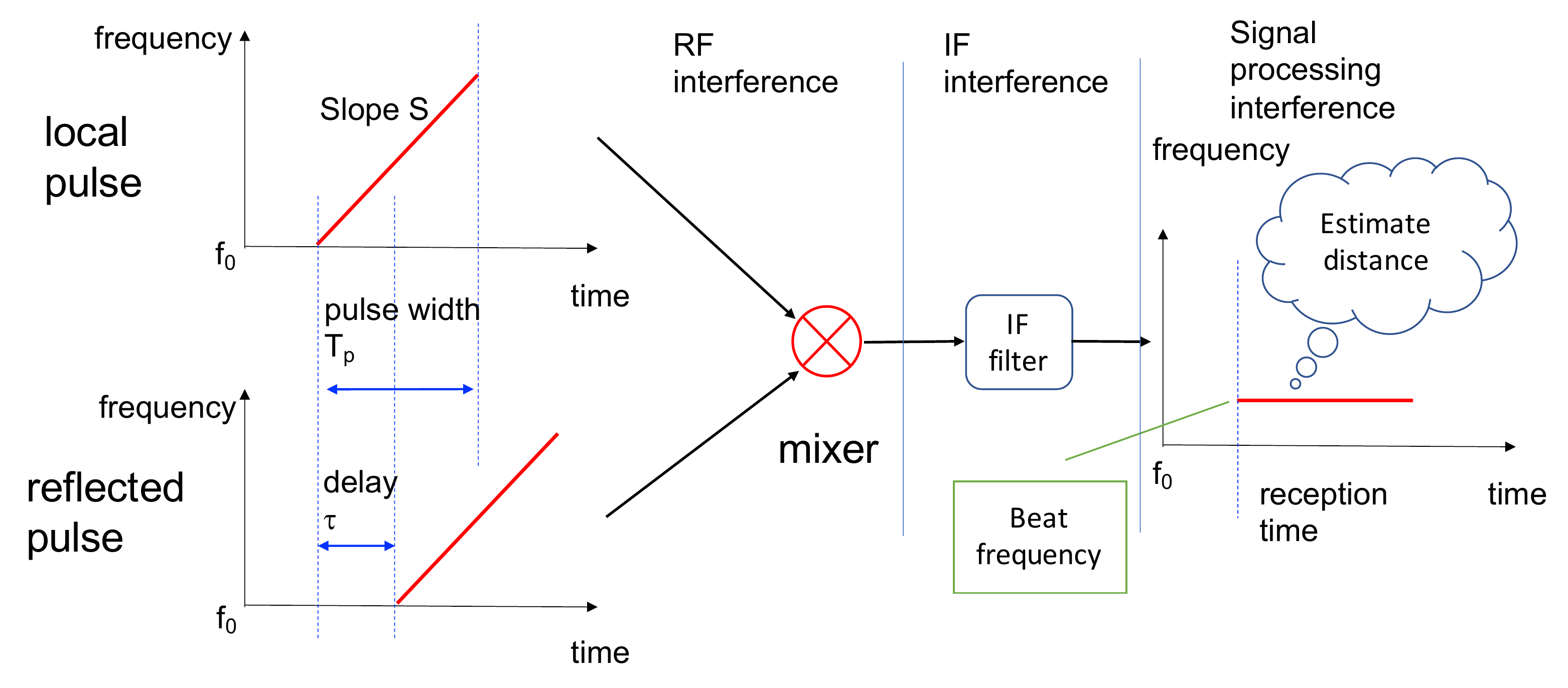}
\caption{An illustration of FMCW radar system}\label{fig:FMCW}
\label{pattern}
\end{figure}

\section{Probabilistic Analysis of Interference}\label{sec:analysis}
In this section, we analyze the probability of interference given the above FMCW radar model. 

\subsection{Hierarchy of Interference}
The focus of this paper is the interference in radar networks. There exist different interferences at different stages of the receiver. Therefore, we consider the following hierarchy of interference, as illustrated in Fig. \ref{fig:FMCW}.

\begin{itemize}
\item RF interference: The interference in the RF band, which exists whenever an interfering pulse overlaps the received pulse.
\item IF interference: The received signal will be mixed with the local pulse and then filtered by an IF filter. Therefore, an RF interference may be filtered out if its frequency is too deviated from the local frequency. If the timing of the interfering pulse is very close to that of the local pulse, the interference may pass through the IF filter, thus causing IF interference. 
\item Decision interference: The IF interference may still be removed by the decision maker. For example, if the time duration of the IF interference is too short (e.g., less than $T_{\min}$ for the radar functionality), it can be determined to be an interference, and thus be discarded. 
\end{itemize}

\subsection{Probability of Interference}

\subsubsection{Interference from Single Interferer}
We first consider a single interferer, say radar $j$, which lays the foundation for the generic case. Since the radar pulses are highly directional, it is very rare for the pulses from one interferer to reach the victim radar (say radar $i$) receiver directly. Therefore, we assume that the victim radar receiver receives the radar pulses via reflection over an object that the interferer is exploring (namely the indirect interference). 

When the pulses from victim radar hits no target, only the interfering radar pulses enter the receiver. According to the mechanism in Section \ref{sec:model}, only when a received pulse arrives within time period $[0,T_p-T_{\min}]$, does the mixing with the local pulse (without loss of generality, within time $[0,T_p]$) and the IF filter output data with beat frequency for ranging. If no special action is taken, the victim radar will claim that there is a target, but with wrong distance and wrong angle. 
Then, we have the following simple proposition for the conditions of wrong ranging due to radar interference.
\begin{proposition}\label{prop:wrong}
Radar $i$ is interfered by radar $j$ if and only if
\begin{eqnarray}\label{eq:conflict1}
\left\lfloor \frac{T_p-T_{\min}+\tau_i-\tau_j-\tau_{ji}}{T_p}\right\rfloor-\left\lfloor \frac{\tau_i-\tau_j-\tau_{ji}}{T_p}\right\rfloor=1,
\end{eqnarray}
where $T_{ij}=\frac{d_{jo}+d_{oi}}{c}$.
\end{proposition}
\begin{proof}
For nonzero output of the IF output with at least duration $T_{\min}$, there exist integers $m$ and $n$ such that 
\begin{eqnarray}\label{eq:conflict2}
0&\leq& \tau_j+nT_p+\tau_{ji}-\tau_i-mT_p\nonumber\\
&\leq&T_p-T_{\min}, 
\end{eqnarray}
which is equivalent to 
\begin{eqnarray}\label{eq:conflict3}
n-m\geq \frac{\tau_i-\tau_j-\tau_{ji}}{T_p},
\end{eqnarray}
and
\begin{eqnarray}\label{eq:conflict4}
n-m\leq \frac{T_p-T_{\min}+\tau_i-\tau_j-\tau_{ji}}{T_p}.
\end{eqnarray}

Then, the solution exists when there exists an integer between the righthand sides of (\ref{eq:conflict3}) and (\ref{eq:conflict4}), which is equivalent to (\ref{eq:conflict1}). This concludes the proof.
\end{proof}

Based on the union bound, a simple corollary of Prop. \ref{prop:wrong} is given below, which shows the probability that the interference causes damage to the victim radar receiver. 
\begin{corollary}\label{cor:wrong}
Suppose that the arrival time of the interfering pulses is uniformly distributed, which is reasonable over all geometric formations of the network. Then, the probability of interference is upper bounded by
\begin{eqnarray}\label{eq:prob}
p_I\leq 
\frac{(N-1)(T_p-T_{\min})}{T_0}.
\end{eqnarray}
\end{corollary}

When the victim radar receives both reflected pulses from its own illumination and interference, the analysis is the same as the the case of no detected target. The conclusion in Prop. \ref{prop:wrong} is still valid. The only difference is that, in the case of existing object, the victim radar receiver has two beat frequencies from the IF filter: one from its own illumination and one from the interference (the probability for the two frequencies to be identical is zero). Therefore, the victim receiver can claim that there are two targets, which implies the existence of interference. However, without further action, the receiver does not know which target is the one it correctly detects. 

\subsubsection{Interference from Many Interferers}

A precise analysis on the interference is difficult, due to the complexity of the geometry of the radars, as well as the randomness in the propagation of radar signal. Therefore, we adopt a simpler model for facilitating the analysis. We assume that each radar monitors only its neighbors with a distance $d_s$, where the subscript $s$ means sensing. At each time, each radar transmitter selects one of its neighbors to sense. It is assumed that for each sensing each radar can always locate the desired neighbor, which is reasonable since it is easy for radar transmitter to predict the mechanical motions of neighbors. During each sensing period, each radar receiver is interfered by the reflections of its neighbors (within a distance of $d_s$) that are being sensed (illuminated) by another radar transmitter. We further assume that $N=\infty$ and the positions of nodes are 2-dimensional Poisson processes with density $\lambda$. Since the the round trip distance is $2d_s$ when a target of distance $d_s$ is sensed, it is reasonable to assume that a transmitter within a distance of $2d_s$ can cause interference, when the reflection loss is omitted. The subsequent conclusions can be easily extended to the case of nonzero reflection loss (which is valid in practical scenarios).

The following proposition justifies the assumption of uniform distribution of received interference pulses in the above analysis, given the assumption that the radar transceivers are not time synchronized.
\begin{lemma}\label{lem:uniform}
When the pulse starting times are independent among the radar transmitters, the starting time of each interfering pulse is uniformly distributed in the corresponding pulse repetition period (namely the period containing the interfering pulse) of the victim radar receiver. 
\end{lemma}
\begin{proof}
Suppose that the starting time of transmission of a certain interfering pulse is $t$ and the propagation time is $t'$ (from the interfering radar transmitter to the victim radar receiver). The statement of the proposition is equivalent to saying that the random variable $x=t+t'\mbox{ mod }(T_p)$ is uniformly distributed in $[0,T_p]$. 

It is easy to see that, due to the assumption of independent starting times, $t\mbox{ mod } (T_p)$ is uniformly distributed in $[0,T_p]$. We claim that the distribution of $t'$ does not affect the distribution of $x=t+t'\mbox{ mod }(T_p)$. This can be proved using the following simple argument, where $x_0$ and $x_0'$ are arbitrary numbers within $[0,T_p]$.
\begin{eqnarray}
p_x(x_0)&=&\int_0^{T_p}\sum_{s':s+s'mod(T_p)=x_0}p_t(s)p_{t'}ds\nonumber\\
&=&\int_0^{T_p}\sum_{s':s+s'mod(T_p)=x_0}p_t(s-(x_0-x_0'))p_{t'}ds\nonumber\\
&=&\int_0^{T_p}\sum_{s':s+s'mod(T_p)=x_0'}p_t(s)p_{t'}ds\nonumber\\
&=&p_x(x_0'),
\end{eqnarray}
namely the probability densities of $x=x_0$ and $x=x_0'$ are the same. This concludes the proof.
\end{proof}

The following proposition shows an upper bound for the probability of being interfered for a radar receiver. The proof is given in Appendix \ref{appx:bound}.
\begin{proposition}\label{prop:bound}
When the positions of radar transceivers are 2-dimensional Poisson distributed with density $\lambda$, the probability for an arbitrary radar receiver to be interfered, which is denoted by $p_I$, is upper bounded by
\begin{eqnarray}
p_I\leq 1-e^{-\left(\frac{T_p-T_{\min}}{T_p}\right)\lambda \pi d_s^2}.
\end{eqnarray}
\end{proposition}

Another metric to characterize the level of radar signal interference is the expected number of interfering signals, denoted by $E[I]$. The higher the expectation is, the more ambiguity the radar receiver faces. The following proposition provides an explicit expression for this expectation. The proof is given in Appendix \ref{appx:proof}.
\begin{proposition}\label{prop:prob}
The expected number of interfering radar signal that a random radar receiver receives is given by
\begin{eqnarray}
E[I]=\frac{1}{2}\lambda d_s^2\left(\frac{T_p-T_{\min}}{T_p}\right).
\end{eqnarray}
\end{proposition}

\section{Detection of Radar Interference}\label{sec:detection}
In this section, we discuss the detection of radar interference. We first analyze the prominent features of radar interference. Then, we use machine learning to classify legitimate signal and interference. 

\subsection{Features for Classification}

\begin{figure}
  \centering
  \includegraphics[scale=0.32]{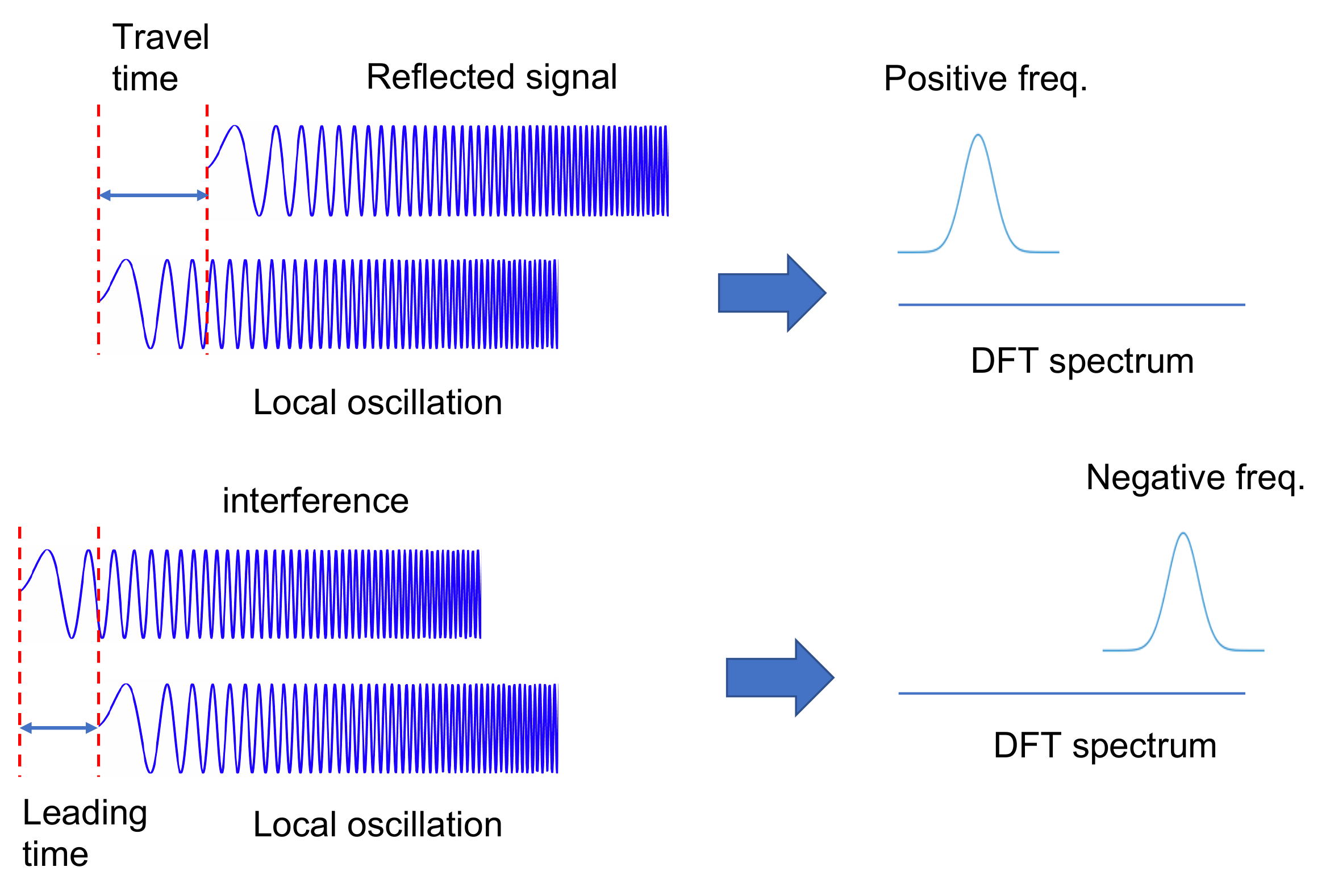}
  \caption{The generation of negative spectrum peak by interference.}\label{fig:ratio}
  \vspace{-0.1in}
\end{figure}

The following two features are used to distinguish the interference and legitimate signal:
\begin{itemize}
\item Power ratio: In legitimate reflected signals with FMCW waveforms, the distance is estimated from the beat frequency, which is proportional to the signal traveling time (thus proportional to the distance). The local oscillation will start before the reflected signal returns. However, when interference is received, the starting time of the interference is random and thus could result in an arrival time leading the local oscillation. This results in a negative beat frequency in the IM output, as illustrated in Fig. \ref{fig:ratio}. Therefore, we use the ratio of the powers in the positive frequency (i.e., the lower half of the DFT spectrum) and negative frequency (i.e., the higher half of the DFT spectrum), to distinguish the interference and legitimate signal. More precisely, the power ratio is defined as
\begin{eqnarray}
R=\frac{\sum_{k=M/2}^{M-1}|X_k|^2}{\sum_{k=0}^{M/2-1}|X_k|^2}.
\end{eqnarray}
\item Doppler: The velocity of the target is estimated via the phase change of successive chirps, namely
\begin{eqnarray}
V_n=\frac{\hat{\phi}_k-\hat{\phi}_{k-1}}{4\pi \lambda T_c},
\end{eqnarray}
where $\lambda$ is the wavelength, and $\hat{\phi}_k$ is the estimated phase of the $k$-th chirp, which is obtained from angle of the frequency spectrum at the peak amplitude. Since the interference is noncoherent with the local oscillation, the phase difference between the successive chirps is random, due to the random initial phase difference. Therefore, we use the average magnitude of velocity estimation $\bar{V}$ as the second feature of interference. 
\end{itemize}

\subsection{Experimental Results}

\subsubsection{Experiment Setup}

\begin{figure}[!t]
\centering
\includegraphics[width=2.6in]{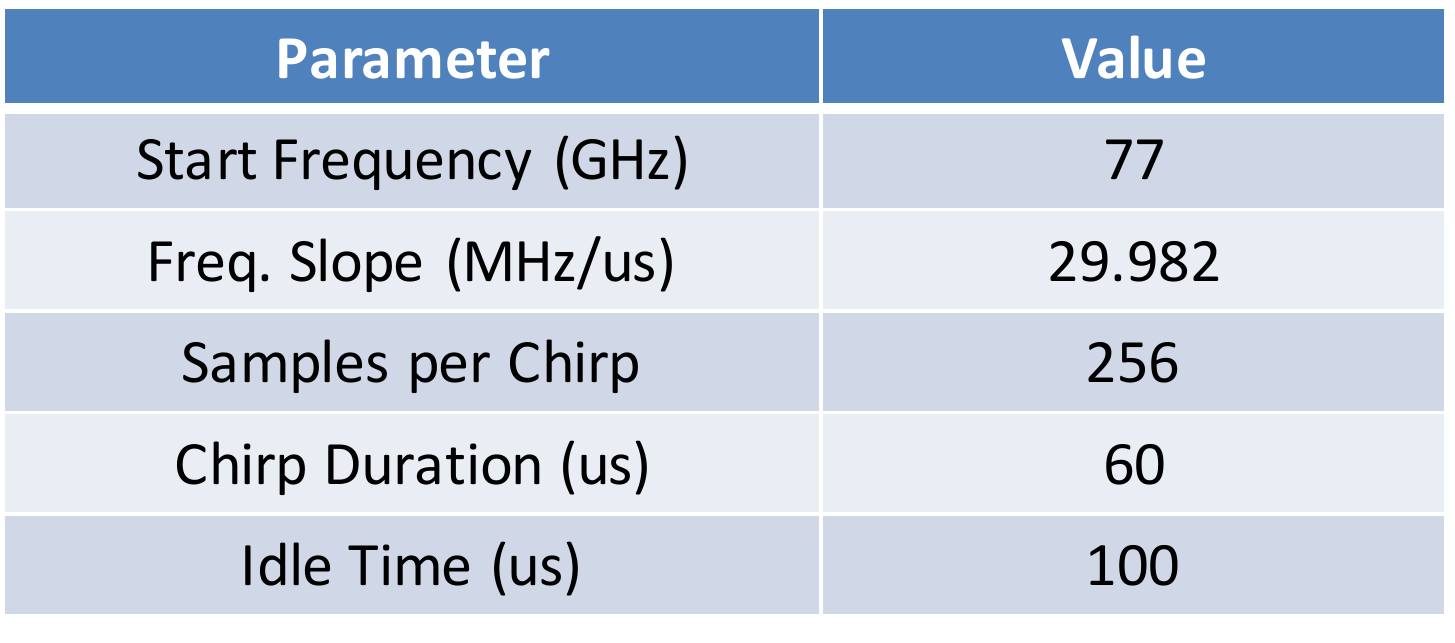}
\caption{Parameters of the experiment}\label{fig:parameters}
\end{figure}

Two mmWave radars (TIAWR1243/1642) in the 77GHz band are used for the experiment. The corresponding parameters are given in Fig. \ref{fig:parameters}. The two radar antennas are placed pointing to each other, thus causing interference. The features are calculated every four chirps. In total 36950 chirps are collected for the interference and 3692 chirps are measured for the legitimate reflected signal. The 2-dimensional features are plotted in Fig. \ref{fig:training}, where the distinction between the interference and signal can be easily observed.

\begin{figure}
  \centering
  \includegraphics[scale=0.36]{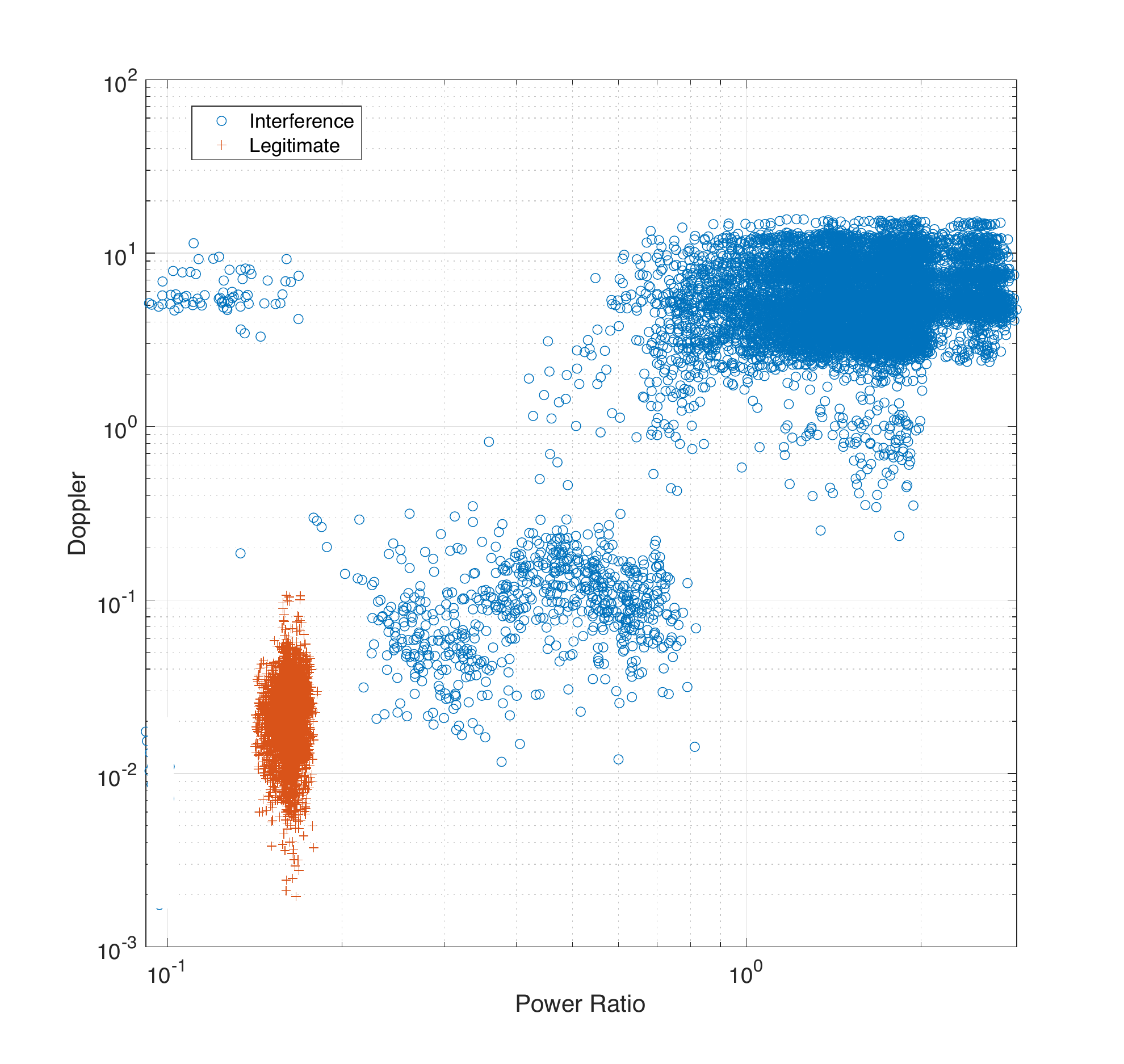}
  \caption{The scatter diagram of training data.}\label{fig:training}
  \vspace{-0.1in}
\end{figure}

\subsubsection{Classification}
The training data is used to train an SVM based classifier. 3692 interference chirps and 3692 signal chirps are collected for testing the SVM classifier. Their features are shown in Fig. \ref{fig:testing}. The trained SVM classifier can distinguish the interference and signal with 100\% correct rate in the testing data set. 

\begin{figure}
  \centering
  \includegraphics[scale=0.4]{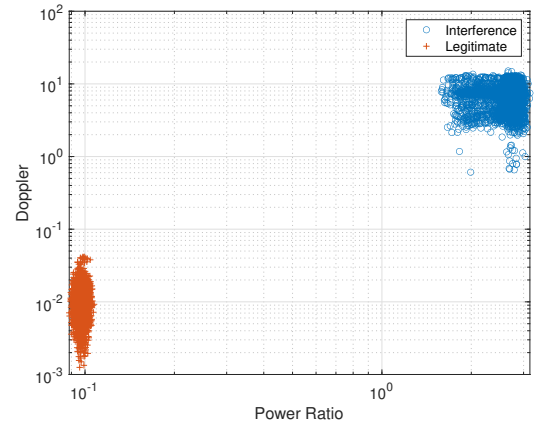}
  \caption{The scatter diagram of testing data.}\label{fig:testing}
  \vspace{-0.1in}
\end{figure}

We also tested the classification without extracting the features. We used the vector consisting of the magnitude spectrum of each chirp and the phase difference of two successive chirps, thus forming a 512-dimensional vector. We used the SVM and linear discrimination analysis (LDA) approaches. In both approaches, the training data set achievers 100\% correct rate, while the testing data attains very bad performance: all the legitimate signal in the testing set has been recognized as interference. This is because that the scenario of signal in the testing data is not observed in the training set. This implies that, without the prior information of the distinction between interference and signal, the black-box classification cannot attain a good performance when no sufficient training data has been collected. 

\section{Radar Interference Management}\label{sec:avoidance}
In this section, we discuss approaches to management the interference of radar impulses. The basic approach is to diversity the radar pulse parameters, such as timing and frequency slope, such that the interference is averaged or avoided. 

\subsection{Time Diversity} 
In the analysis of the impact of interference, we notice that, if an IF interference occurs and the two radar transceivers share the same parameters, it will always exist unless they change beam directions. This is due to the common pulse period, such that, if the time offset of pulses between different nodes is small (which is necessary for an IF interference), it will always be small. To address this problem, we propose to use different pulse periods for different nodes, namely $T_p^1$, ..., $T_p^N$. Then, even if an interfering pulse is very close to a legitimate pulse and thus causes an IF interference, the subsequent pulses will run out of beat and thus remove the IF interference, as illustrated in Fig. \ref{fig:beat}. 

\begin{figure}[!t]
\centering
\includegraphics[width=3in]{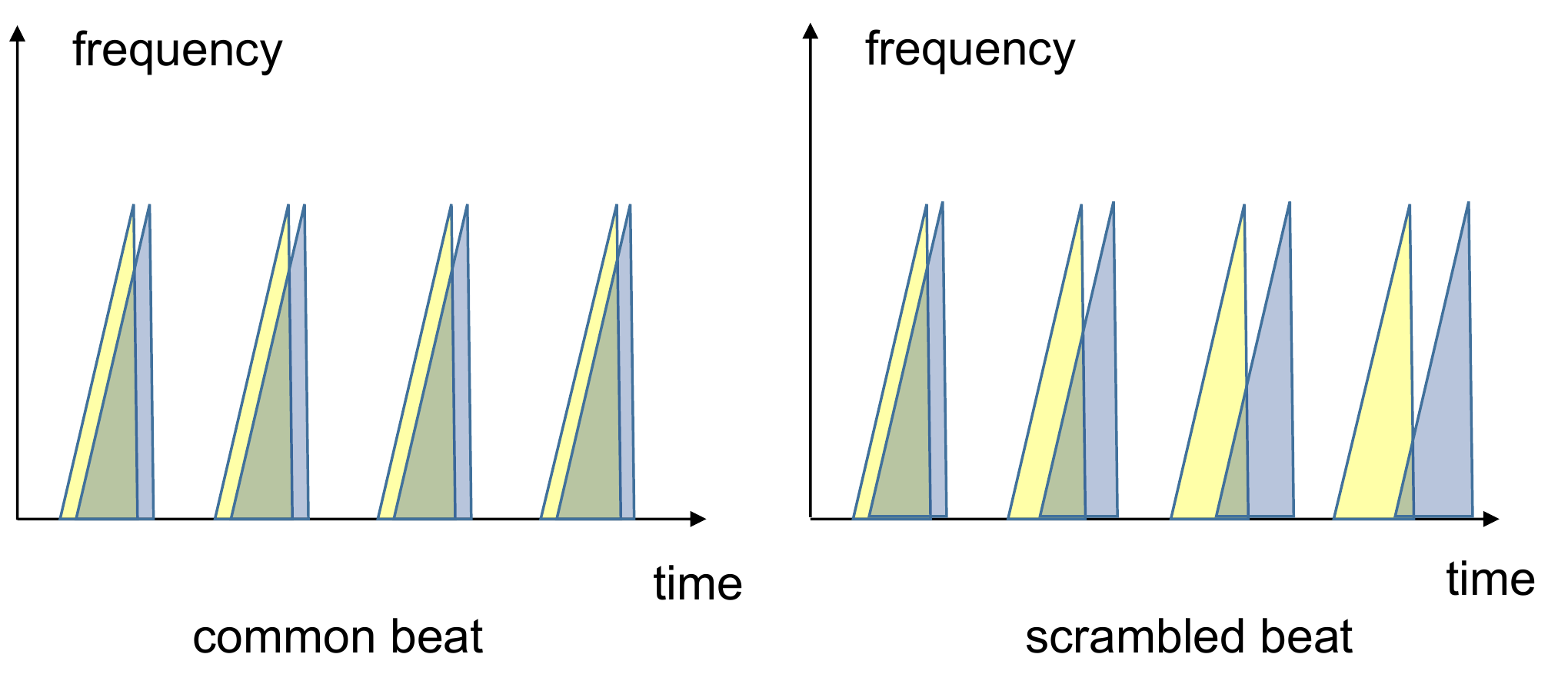}
\caption{An illustration of synchronized and scrambled beats of pulses}\label{fig:beat}
\label{pattern}
\end{figure}

\subsection{Slope Diversity}

\begin{figure}[!t]
\centering
\includegraphics[width=2.7in]{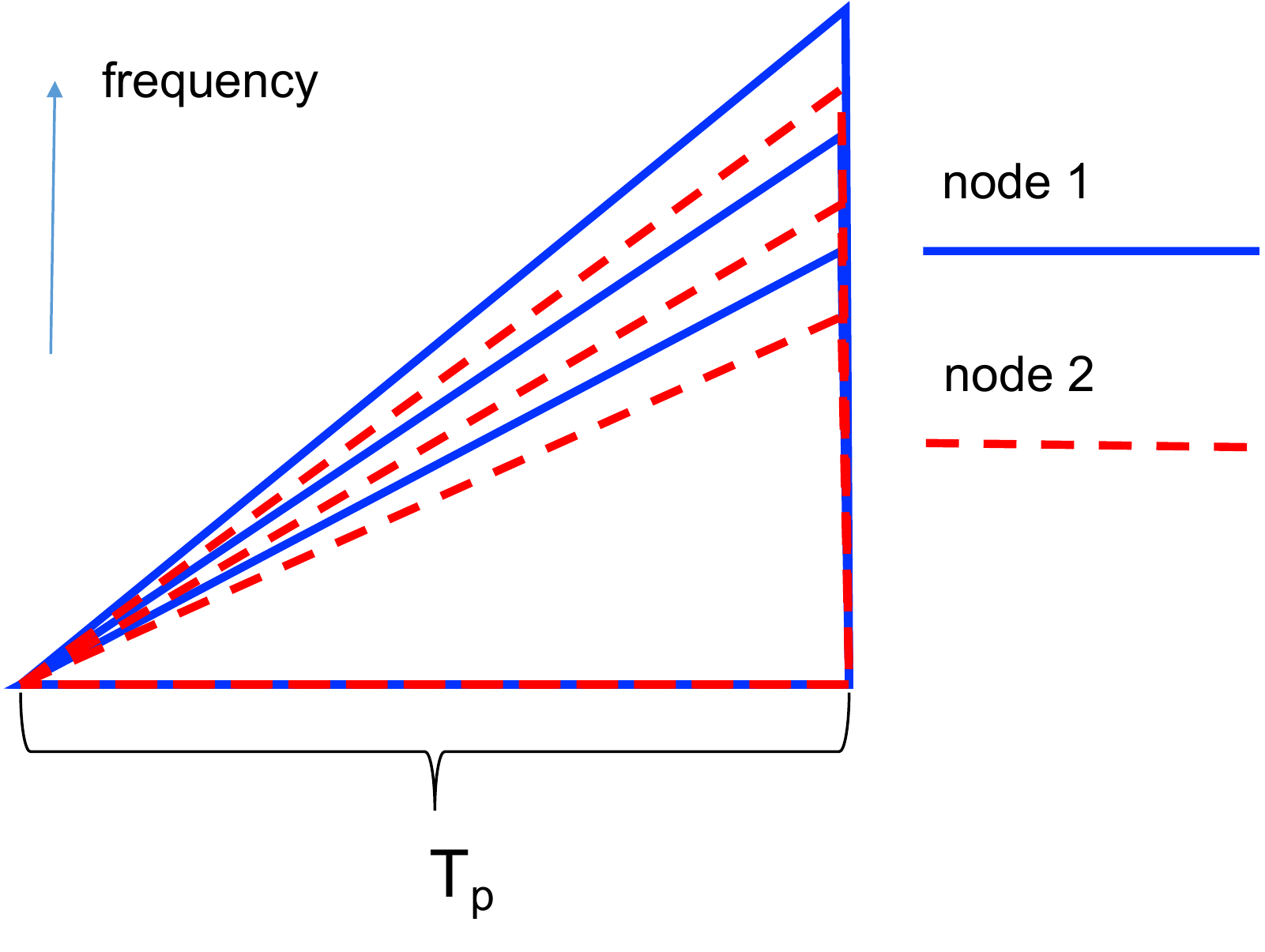}
\caption{An illustration of slope diversity for two nodes}\label{fig:diverse}
\label{pattern}
\end{figure}

\begin{figure}[!t]
\centering
\includegraphics[width=2.3in]{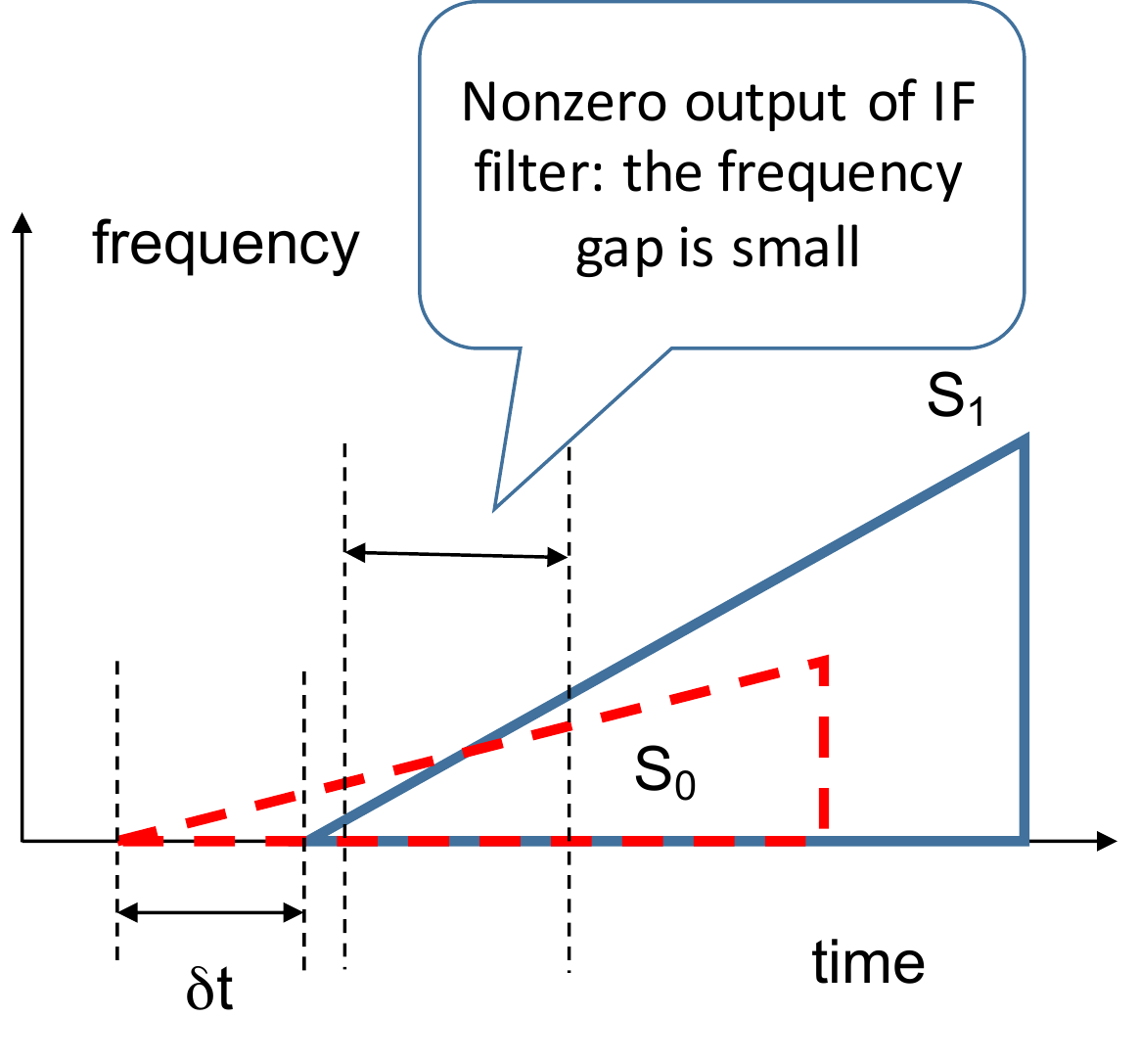}
\caption{Mixing of pulses with different slopes}\label{fig:gap}
\label{pattern}
\end{figure}

To avoid interference, we propose the diversity of frequency slope $S$, namely different radars use different values of $S$. The rationale is illustrated in Fig. \ref{fig:diverse}: when a radar receiver with $S_1$ is interfered by a radar with $S_2$, the frequency of the mixer output changes in an affine manner, namely $f(t)=(S_1-S_2)(t-t_0)+f(t_0)$, where $t_0$ is the beginning time of the interference. Due to the large values of $S_1$ and $S_2$ (and thus a large $|S_1-S_2|$ in a probabilistic manner), the beat frequency $f(t)$ increases very fast and leaves the passband of the IF filter quickly. Therefore, the time duration of interference will be very short and can be easily removed, thus reducing the radar interference.

\section{Leveraging Radar Interference}\label{sec:leverage}
In this section, we discuss how to leverage the information brought by the radar interference. For simplicity, we assume that the radar receiver can perfectly distinguish radar interference from its own radar signal. We further assume that the radar receiver can estimate the direction of the interference with an open angle of uncertainty, namely the source of the interference (another radar transmitter for direct interference, and a reflector for indirect interference). 

\subsection{Information Fusion}
We assume that nearby nodes can communicate with each other (e.g., via V2V communication protocols), or send information to a fusion center (e.g., to a roadside computing device via V2I infrastructure). For simplicity, we assume that a roadside base station receives reports of interference from multiple vehicles, which consist of the information of time, angle and amplitude of the interference. When there is no base station, the vehicles can exchange information and carry out the inference based on the information of interference. Hence, the major challenge is how to extract information from the interference report.

\subsection{World Lines Representations}
We leverage the concept of world line from relativity to describe the dynamics of the mobile radar transceivers, as illustrated in Fig. \ref{fig:space_time}, where the reference radar transceiver records two events of interference. The trajectory (world line) of each mobile radar transceiver is denoted by $(t,\mathbf{x}(t))\in \mathbb{R}\times \mathbb{R}^2$. The event of receiving an interference is a point in the space-time.

\begin{figure}
  \centering
  \includegraphics[scale=0.45]{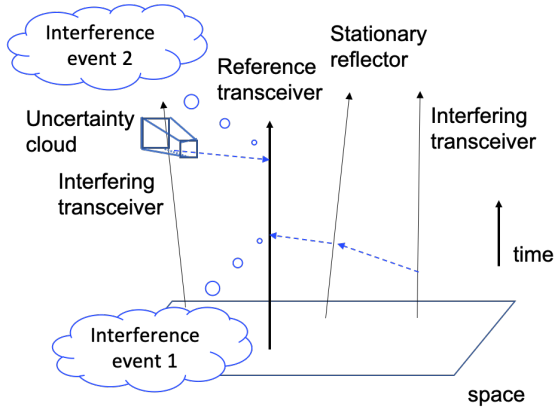}
  \caption{An illustration of world lines and events in the space-time.}\label{fig:space_time}
  \vspace{-0.1in}
\end{figure}

Consider a radar receiver, which detects an interference. It estimates the following information for the interference source:
\begin{itemize}
\item Direction: Using its antenna array, it provides a range of the directional angle, whose set is denoted by ${\Omega}$.
\item Distance: Using the received power strength, it estimates the range of distance, which is denoted by $\mathcal{D}$.
\item Timing: Using the distance estimation, light speed and local clock, it estimates the range $\mathcal{T}$ of the starting time for the interference (direct transmission or reflection). 
\end{itemize}
Therefore, the receiver is able to determine the space-time range ${\Omega} \times \mathcal{D} \times \mathcal{T}$. We call such regions uncertainty clouds. 

\subsection{Target Detection and Tracking}

\begin{figure}
  \centering
  \includegraphics[scale=0.45]{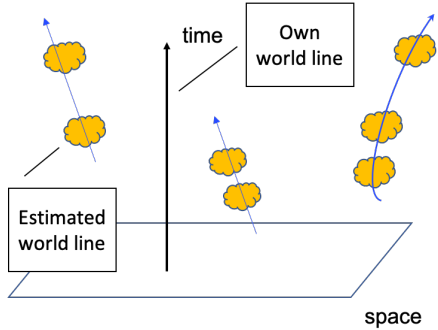}
  \caption{An illustration of clustering the clouds.}\label{fig:clouds}
  \vspace{-0.1in}
\end{figure}

The goal of the reconstruction of target, given the interference record, is to cluster the clouds, as illustrated in Fig. \ref{fig:clouds}. Roughly speaking, if multiple clouds can be transversed by a line (or a curve with small curvature), we can consider these clouds as being generated by a moving target, where the curve is the world line of the target. Then, we are facing the tradeoff between model complexity and precision:
\begin{itemize}
\item Model complexity: The reconstruction precision is optimized if each cloud is considered as a target. However, there will be too many targets, thus making the model too complex.
\item Precision: Using a single curve to pass all clouds will be the simplest model. However, when the curvature is bounded, this yields substantial errors. 
\end{itemize}
Besides the above tradeoff, we also need to consider the practical constraint: the target cannot move beyond speed limit ($v_m$). Note that the speed is given by the cotangent of the intersection angle of the world line and the space plane (here 2-dimensional space is considered). Based on the above considerations, we propose the target detection and tracking algorithm in Procedure \ref{alg:search}.

\begin{algorithm}[H]
	\caption{}\label{alg:search}
	\begin{algorithmic}[1]
		\STATE{Calculate the maximal angle $\theta=\cot^{-1}(v_m)$.}
		\STATE{Collect all clouds into a set $\Phi$.}
		\STATE{Set the set of world lines, $L$, to be empty }
		\WHILE{$\Phi$ is nonempty}
   		        \STATE{Select the earliest cloud $\phi$.}
		        \STATE{Search the line passing $\phi$ and the maximum number of other clouds.}
		        \STATE{Remove these clouds from $\Phi$.}
		        \STATE{Add the line to $W$.}
		\ENDWHILE
		\STATE{Output the world lines in $W$, as well as the corresponding clouds.}
	\end{algorithmic}
\end{algorithm}

\section{Multiuser Sensing}\label{sec:multiuser}
In this section, we propose a framework called multiuser sensing, in order to identify and decouple interferences, given the knowledge of known radar waveform(s). The multiuser sensing is a counterpart of celebrated multiuser detection in communications \cite{Verdu1998}. Compared with the celebrated multiuser detection theory \cite{Verdu1998} in communications, which usually assumes time synchronicity, the proposed multiuser sensing framework has a similar linear representation but focuses on asynchronous signals. The time asynchronicity is essential, since the time offset of signal provides information on the target distance, and thus requires processing in the $Z$-transform domain.

\subsection{Formulation}
We formulate the radar interference as the following linear superposition signal model:
\begin{eqnarray}\label{eq:linear}
\mathbf{y}(z)=\mathbf{A}(z)\mathbf{x}(z)+\mathbf{w}(z),
\end{eqnarray}
or equivalently the following elementwise form:
\begin{eqnarray}
y_n(z)=\sum_{k=1}^K A_{nk}(z)x_k(z)+w_n(z),n=1,...,K,
\end{eqnarray}
where $\mathbf{A}$ is a $K\times K$ matrix which consists of the channel gains (when $A_{nk}\neq 0$ and $k\neq n$, user $k$ interferes user $n$; $A_{nn}$ could be zero if the signal of user $n$ is not reflected back to the transmitter), $K$ is the total number of users in the area, which is unknown in advance, and both $\mathbf{y}$ and $\mathbf{x}$ are $K$-dimensional vectors, denoting the received and transmitted signals, respectively, and $\mathbf{w}$ is the vector of thermal noise. Note that the elements in $\mathbf{A}$, $\mathbf{y}$ and $\mathbf{x}$ are polynomials of $z$, namely the $Z$-transforms of the sampled signals and channel propagations, where $z^{-1}$ means unit delay. In particular the elements in $\mathbf{A}$ are monomials, since the propagation of one radar signal is determined by the magnitude and delay: when $A_{nk}=gz^{-\tau}$, the channel amplitude gain is $g$ and the time delay of signal is $\tau$. A simple example is $\mathbf{x}=(1+2z^{-1}+z^{-2};2+z^{-1}+z^{-1})$ and $\mathbf{A}=\begin{pmatrix}z^{-1},0.5z^{-3}\\z^{-3},5z^{-1}\end{pmatrix}$, which indicates (a) the sampled waveforms of users 1 and 2 are $(1,2,1)$ and $(2,1,1)$; (b) the time delays of reflected signals are 1 sample period, while the time delays of interference are both 3. 

The goal of multiuser sensing algorithm is for each user $n$, given the observation $y_n(z)$, to determine whether the received signal is interfered, to decouple the signal and interference, and to estimate the target if there is any. The challenge is that the matrix $\mathbf{A}$, even its dimensions (how many users in the region), is unknown in advance. Fortunately, the knowledge that the users can leverage includes (a) the waveforms in $\mathbf{x}$ are assumed to be known; (b) the polynomial order (namely the time delay) of each elements in $\mathbf{A}$ cannot be too large since it is limited by the distance between the users; (c) in typical situations, there are very few off-diagonal elements in $\mathbf{A}$, namely a user is not simultaneously interfered by many interferers. Given the above conditions, each user $n$ estimates the corresponding row $\mathbf{a}_n=(A_{n,1},A_{n,2},...,A_{n,K})$, given the observation $y_n$ and the known waveforms $\mathbf{x}$. Note that, although the linear form in (\ref{eq:linear}) is similar to that of multiuser detection, the goal and features are substantially different, which results in significantly different methodologies. The case of single-user decision will be extended to the case of joint decision; e.g., users 1, 2 and 3 can communicate with each other, then they can make the decision based on $y_1(z)$, $y_2(z)$ and $y_3(z)$.

\subsection{Time Domain Algorithm}
We begin with a time-domain algorithm for the multiuser sensing, which is more suitable for radar waveforms designed in the time domain, e.g., the pseudorandom waveform (noise radar).

\subsubsection{Single Waveform}
When all radar users have the same waveform $x_0(t)$ and the received signal is noise free, the row vector $\mathbf{a}_n$ of user $n$ can be obtained from the polynomial division:
\begin{eqnarray}
\frac{y_n(t)}{x_0(t)}=\sum_{j=1}^m a_j z^{-n_j},
\end{eqnarray} 
where $a_j$ is the amplitude of signal and $n_j$ is the time delay (in the unit of sampling period). This polynomial division decouples the interference and signal. The terms $a_jz^{-n_j}$ with abnormal $n_j$ (e.g., being too large) will be claimed as interference and canceled from the received signal. The remainder of the signal, if there is any, will be considered as the legitimate signal. 

\subsubsection{Multiple Waveforms}
We now extend the single-waveform case to multiple possible waveforms $\{p_m(z)\}_{m=1,...,M}$ (e.g., pulse radar, FMCW, OFDM, et al). Then, the received signal $y(z)$ is in the ideal generated by $\{p_m(z)\}_{m=1,...,M}$. The multiuser sensing is converted to the problem of decomposition:
\begin{eqnarray}\label{eq:decomp}
y_n(t)=\sum_{m=1}^M a_m(z) p_m(z)+R(z),
\end{eqnarray}
where the polynomial coefficient $a_m(z)$ determines the amplitudes and delays of users using $p_m(z)$, and $R(z)$ is the residue due to noise or waveform imperfection. A major challenge is that the decomposition is not unique. For example, consider $M=2$ and $y(t)=p_1(z)+p_2(z)$. Assume that $p_1(z)$ and $p_2(z)$ are coprime, which implies the existence of $a(z)$ and $b(z)$ such that $a(z)p_1(z)+b(z)p_2(z)=1$. Therefore, $y(t)$ can also be decomposed to $y(t)=c(z)p_1(z)+d(z)p_2(z)$, where $c(z)=y(z)a(z)$ and $d(z)=y(z)b(z)$. We will leverage the prior information on the orders of the elements in $\mathbf{A}(z)$ to exclude improper decompositions; e.g., the decomposition $c(z)p_1(z)+d(z)p_2(z)$ in the above example can be excluded because the orders of $c(z)$ and $d(z)$ will be too large. 

For the decomposition in (\ref{eq:decomp}), we use the following polynomial division algorithm (Page 47 in \cite{Smith2014}), where $LT(f)$ is the leading term of polynomial $f$, e.g., $LT(2z^{-2}+3z^{-1}+4)=2z^{-2}$. In the output of Algorithm \ref{alg:division}, the coefficients $\{a_i\}_{i=1,...,M}$ indicates the magnitude and delay of signals with different waveforms, and the residue $R$ indicates the remainder noise and the precision of the decoupling. 

\begin{algorithm}[H]
	\caption{}\label{alg:division}
	\begin{algorithmic}[1]
		\STATE{Given the received signal $y(z)$ and waveforms $\{p_1(z),...,p_M(z)\}$}
		\STATE{Set $a_i=0$, $i=1,...,M$, $R=0$ and $g=y$.}
		\WHILE{$g\neq 0$}
   		        \STATE{Set $Matched=false$.}
		        \FOR{i=1,...,M}
		                 \IF{$LT(p_i)$ divides $LT(g)$}
        		                     \STATE{Set $h=LT(g)/LT(p_i)$.}
		                     \STATE{Set $a_i=a_i+h$.}
		                     \STATE{Set $g=g-p_ih$.}
		                     \STATE{Set $Matched=true$; break.}
		                 \ENDIF
		        \ENDFOR
		        \IF{not $Matched$}
		              \STATE{Set $R=R+LT(g)$.}
		              \STATE{Set $g=g-LT(g)$.}
		        \ENDIF
		\ENDWHILE
		\STATE{Return $\{a_i\}_{i=1,...,M}$ and $R$.}
	\end{algorithmic}
\end{algorithm}

\subsection{Frequency Domain Algorithm}
When the radar waveforms are more featured by their frequency spectrums, such as the FMCW or OFDM radar systems, the time-domain algorithms will be difficult, since prohibitively many time samples are needed (e.g., for the rapidly changing FMCW waveforms), thus making the polynomial processing in the time-domain algorithm difficult. In such a scenario, we propose to use the frequency-domain approach. Again, we start from the case of identical radar waveforms. The key observation is that, when $z=e^{j\Omega}$, $y_n(z)$ and $x_0(z)$ equal the discrete-time Fourier transform (DTFT), at frequency $\Omega$, of received signal and radar waveform, respectively. Meanwhile, given $z=e^{j\Omega}$, $a_n(z)=\sum_{k=1}^K A_{nk}(z)$ equals the frequency response at frequency $\Omega$, if (\ref{eq:linear}) is considered as a linear system. Therefore, one can compute the frequency spectrum of received signal and radar waveform, and thus obtain the frequency response of the linear system characterized by $\mathbf{A}$. An inverse DTFT (IDTF) generates $a_n(z)$ and thus the elements in $\mathbf{A}$. 

To this end, we consider the Fourier transform of (\ref{eq:decomp}). For simplicity in concept and derivation, we consider the continuous-time Fourier transform (CTFT), which is given by
\begin{eqnarray}\label{eq:decomp2}
Y_n(jw)=\sum_{m=1}^M A_m(jw) P_m(jw)+R(jw),
\end{eqnarray}
where
\begin{eqnarray}
A_m(jw)=\sum_{l=1}^{r_l} a_{ml}e^{-jw\tau_{ml}},
\end{eqnarray}
where $a_{ml}>0$ is the amplitude of the $l$-th branch and $e^{-jw\tau_{ml}}$ is the corresponding phase. 

\section{Numerical Simulations}\label{sec:numerical}
In this section, we provide the numerical simulation results that demonstrate the validity of the proposed algorithms for radar interference. 

\subsection{Simulation Setup}
We consider a crossroad as illustrated in Fig. \ref{fig:crossing}. The two orthogonal roads are each 600 meters long, and the speed limit is 30mph. The vehicles take turn to go through the crossroad. They may decelerate to avoid collision, and resume acceleration when the road ahead is cleared. We consider only the radar in the front of each vehicle. The radiation pattern is determined by that of a 16-element array. Only the direct interferences from the vehicles in the opposite lanes are considered. Random timing offsets are assumed for the radar pulses at different vehicles, while the time drift due to the imperfection of oscillator is not taken into account. Each radar receiver discards the received radar pulses that result in negative distance or distances beyond the maximum range.

\begin{figure}
  \centering
  \includegraphics[scale=0.45]{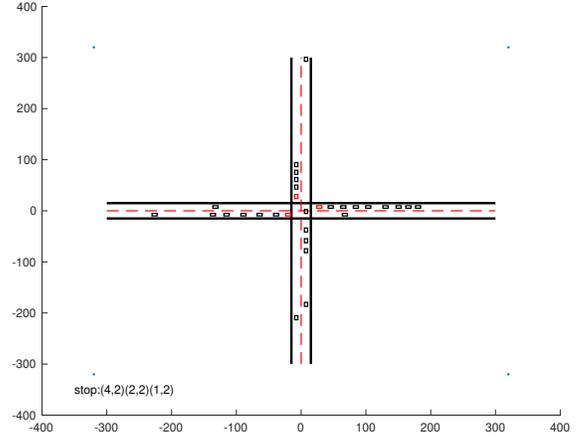}
  \caption{Crossroad setup in the simulation.}\label{fig:crossing}
  \vspace{-0.1in}
\end{figure}

We assume FMCW radar is used, where we use the parameters of TI AWR1243BOOST Pack (76-81GHz). We assume that the chirp duration is 60us and the pulse period is 100us. The frequency increasing slope is 29,982GHz/s. 

\subsection{Statistics of Radar Interference}

The simulation is carried out for a time period of 7 hours, during which 4051 vehicles passed the crossroad. The traffic follows a Poisson process, where the average emergence rate of new vehicles is 0.64 vehicles per second. In Fig. \ref{fig:num_amp}, the histograms of times of radar interference events and the amplitudes are shown. We observe that, during the journeys of most vehicles, there is only one interference event. The probability of more than one interference events is much smaller. The amplitude of interference is more like a bell shape, centered at around -110dBW. 

\begin{figure}
  \centering
  \includegraphics[scale=0.45]{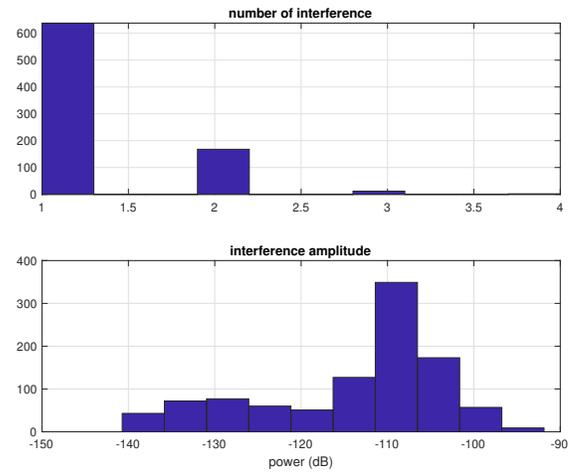}
  \caption{Histograms of times and amplitudes (in dBW) of interference.}\label{fig:num_amp}
  \vspace{-0.1in}
\end{figure}

The histograms of interference time duration and inter-interference interval are shown in Fig. \ref{fig:dur_per}. We observe that both histograms are similar to the exponential distributions. Using the Lilliefors test at the 5\% significance level, the distribution of the interference duration is shown not to be exponential. Meanwhile, the inter-interference interval is shown to be exponentially distributed. We plotted the CDFs of both the empirical data and exponential distribution, for both variables, in Fig. \ref{fig:CDF}. We observe that, while the empirical and exponential CDF curves are very close to each other, the interference duration has a slightly longer tail than the exponential distribution, which implies possibly longer interference duration due to the waiting time at the crossroad.

\begin{figure}
  \centering
  \includegraphics[scale=0.45]{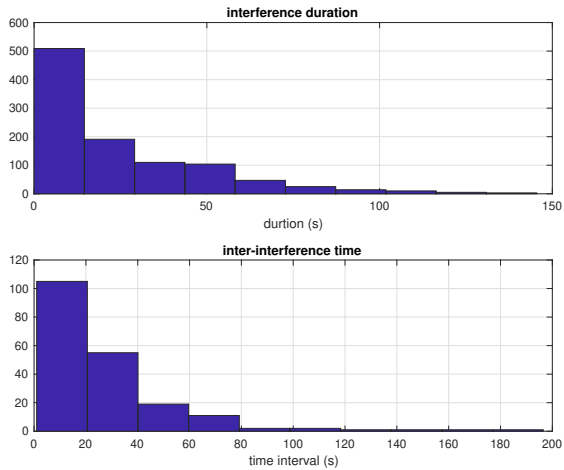}
  \caption{Histograms of interference duration and inter-interference interval.}\label{fig:dur_per}
  \vspace{-0.1in}
\end{figure}

\begin{figure}
  \centering
  \includegraphics[scale=0.45]{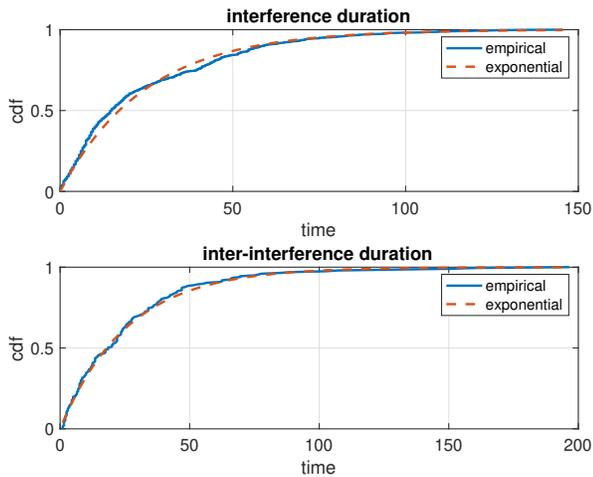}
  \caption{Comparison of CDFs of empirical data and exponential distribution.}\label{fig:CDF}
  \vspace{-0.1in}
\end{figure}

\begin{figure}
  \centering
  \includegraphics[scale=0.45]{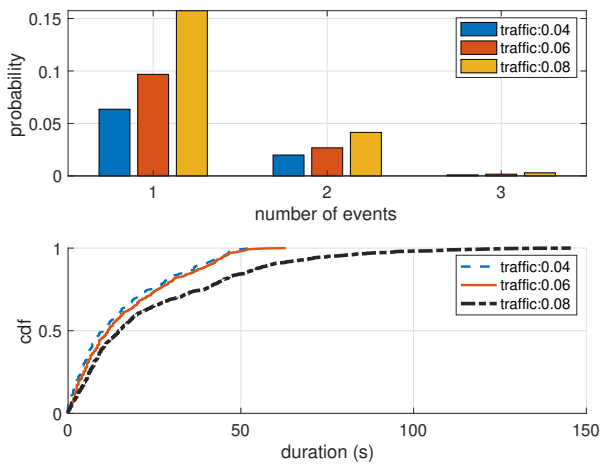}
  \caption{Comparison of different traffic densities.}\label{fig:density}
  \vspace{-0.1in}
\end{figure}

\subsection{Radar Interference Avoidance}

\begin{figure}
  \centering
  \includegraphics[scale=0.4]{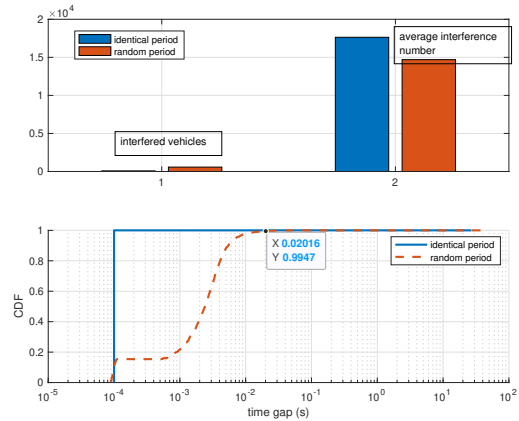}
  \caption{Performance of interference avoidance by time diversity.}\label{fig:flexible}
  \vspace{-0.1in}
\end{figure}

The performance of interference avoidance using time diversity, compared with identical radar impulse period, is shown in Fig. \ref{fig:flexible}. For the beat scramble, we consider a uniform distribution between 0.9 and 1.1 times of the standard impulse period (100us). In both simulations for the time diversity and identical period, 601 vehicles passed the crossroad, only 77 vehicles experience interference in the identical period case, in a sharp contrast to 577 in the time diversity case. Therefore, much more vehicles suffer from interference when time diversity is used. However, this does not mean that the time diversity is a worse design. As shown in the upper part of Fig. \ref{fig:flexible}, the average number of interference (averaged over all passing vehicles) is slightly less for the time diversity case than the identical period case. In the CDF curves of inter-interference-impulse intervals shown in the lower part of Fig. \ref{fig:flexible}, the intervals in the time diversity case are statistically greater, while the intervals of the identical period are mostly a single period (since once interference occurs it will be for every impulse within a time period of seconds). Therefore, although the time diversity cannot substantially decrease the number of interference impulses and even increases the number of interfered vehicles, the interference is substantially scattered among vehicles and time. Similarly to communications, scattered interference can be more easily isolated and mitigated, thus benefiting the performance of radar sensing.

\begin{figure}
  \centering
  \includegraphics[scale=0.4]{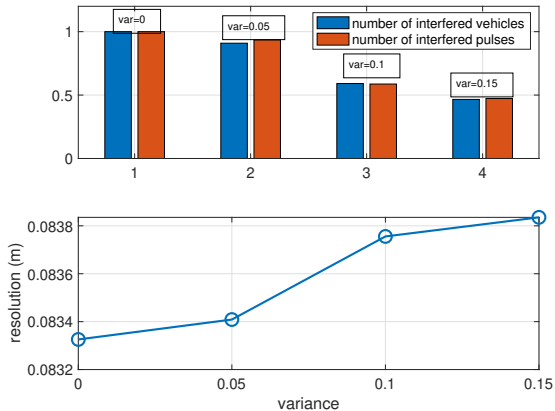}
  \caption{Performance of interference avoidance by slope diversity.}\label{fig:sdiversity}
  \vspace{-0.1in}
\end{figure}

The performance of interference avoidance using slope diversity is shown in Fig. \ref{fig:sdiversity}, where the IF output with time duration less than 80\% of the standard pulse duration is declined as an interference. The slope is randomly and uniformly selected with standard deviation ranging from 0 to 0.15. The numbers of interfered vehicles and interfered radar pulses are shown in the upper part of Fig. \ref{fig:sdiversity}, where the numbers are normalized by the ones of zero variance (thus no slope diversity). We observe that the interference is substantially reduced, particularly to less than half when the standard deviation is 0.15. The lower part of Fig. \ref{fig:sdiversity} shows the degradation of average resolution of ranging distance when the variance of slope is increased. Note that the distance resolution is given by $\frac{c}{2ST_c}$. We observe that the increase of distance resolution is marginal (in the magnitude of millimeters). Indeed, the bottleneck of FMCW radar, in the setup of the simulations, is the sampling rate at the receiver. Therefore, the increase of slope variance almost does not degrade the ranging performance, while substantially reducing the radar interference. 

\subsection{Leveraging Interference}
\begin{figure}
  \centering
  \includegraphics[scale=0.4]{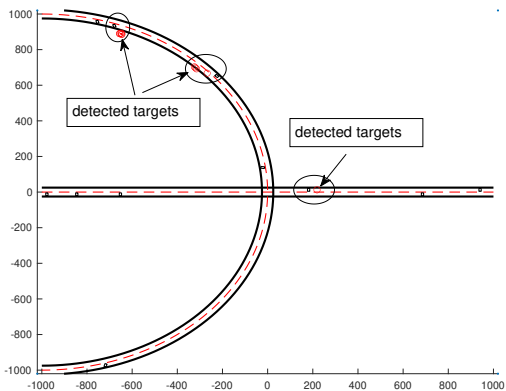}
  \caption{Examples of detected targets by leveraging interference.}\label{fig:detected_target}
  \vspace{-0.1in}
\end{figure}

We implemented the simulation in Matlab for the details of the algorithm in Procedure \ref{alg:search}. In particular, we consider two freeways with a bridge (thus without a stop sign or traffic light), as illustrated in Fig. \ref{fig:detected_target}. The average traffic intensity is 0.06 new vehicles per second. Slight random perturbations are considered in the velocities of the vehicles. Several examples of detected targets, by leveraging the interference, are marked in Fig. \ref{fig:detected_target}. We observe that the detected targets are close to real vehicles. The histogram of position estimation errors of the targets detected by interference, within a time span of 10 seconds, is shown in Fig. \ref{fig:hist_error}. We observe that, for the majority of the cases, the position estimation errors are tens of meters, which is much greater than that of radar ranging (which could be within a meter). The main reason for the errors is the ambiguity in the beam directions due to the beam open angle. However, the error is tolerable and the estimation is useful for spotting long-distance (around 1km) vehicles. We also observe that there are some cases of significant estimation errors. Due to the numerical analysis of the errors, we found that the significant error is due to the blurring of image due to radar receiver movement. This source of error will be addressed in our future research. 

\begin{figure}
  \centering
  \includegraphics[scale=0.4]{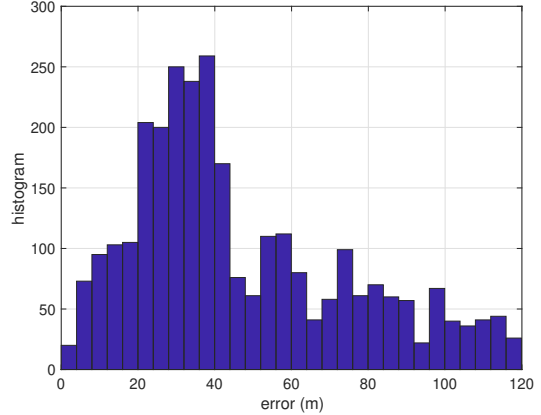}
  \caption{Histogram of detection errors of targets using interference.}\label{fig:hist_error}
  \vspace{-0.1in}
\end{figure}

\section{Conclusions}\label{sec:conclusions}
In this paper, we have discussed the mitigation and leverage of interference in mobile radar networks. Different levels of interference management have been proposed and analyzed, including the interference avoidance/average, detection, decoupling and leverage. The validity of the proposed algorithms has been demonstrated using numerical simulations, in the context of radar sensing in vehicular networks. 

\appendices
\section{Proof of Prop. \ref{prop:bound}}\label{appx:bound}
\begin{proof}
According to the conclusion in Lemma \ref{lem:uniform}, if a node within the range $d_s$ of the radar receiver being studied, the probability that it causes interference is given by (\ref{eq:prob}). We further assume that each node within the range $[0,d_s]$ to the receiver under study is being illuminated by a certain radar transmitter, thus causing interference due to the system model. Then, we have
\begin{eqnarray}
p_I&=&\sum_{n=0}^\infty \frac{e^{-\lambda \pi d_s^2}(\lambda \pi d_s^2)^n}{n!}P(interfering)\nonumber\\
&\leq &\sum_{n=0}^\infty \frac{e^{-\lambda \pi d_s^2}(\lambda \pi d_s^2)^n}{n!}\left(1-\left(1-\frac{T_p-T_{\min}}{T_p}\right)^n\right)\nonumber\\
&=&1-\sum_{n=0}^\infty \frac{e^{-\lambda \pi d_s^2}(\lambda \pi d_s^2)^n}{n!}\left(1-\frac{T_p-T_{\min}}{T_p}\right)^n\nonumber\\
&=&1-e^{-\lambda \pi d_s^2}e^{\left(1-\frac{T_p-T_{\min}}{T_p}\right)\lambda \pi d_s^2}\nonumber\\
&=&1-e^{-\left(\frac{T_p-T_{\min}}{T_p}\right)\lambda \pi d_s^2}.
\end{eqnarray}
This concludes the proof.
\end{proof}

\section{Proof of Prop. \ref{prop:prob}}\label{appx:proof}
\begin{proof}
All possibly interfering radar transmitters are located within a sphere centered at the radar receiver under study and of radius $2d_s$. The number $I$ of interfering radar signals, as a random variable, can be written as
\begin{eqnarray}
I=\sum_{i}^{M}I_i,
\end{eqnarray}
where $M$ is the number of radar transmitters within the radius $2d_s$ and $I_i$ is the indicator function for the event that the $i$-th radar transmitter interferes the radar receiver. Taking expectation, we have
\begin{eqnarray}
E[I]=E\left[\sum_{i}^{M}E[I_i|M]\right],
\end{eqnarray}
where $E[I_i|M]$ is the conditional expectation given $M$. Due to the uniform distribution of the node positions, $E[I_i|M]$ is independent of $M$. Moreover, $E[I_i|M]=E[I_j|M]$ due to the symmetry of positions. Therefore we have
\begin{eqnarray}
E[I]=E[M]E[I_i]=2\lambda \pi d_s^2E[I_i].
\end{eqnarray}

\begin{figure}[!t]
\centering
\includegraphics[width=2.2in]{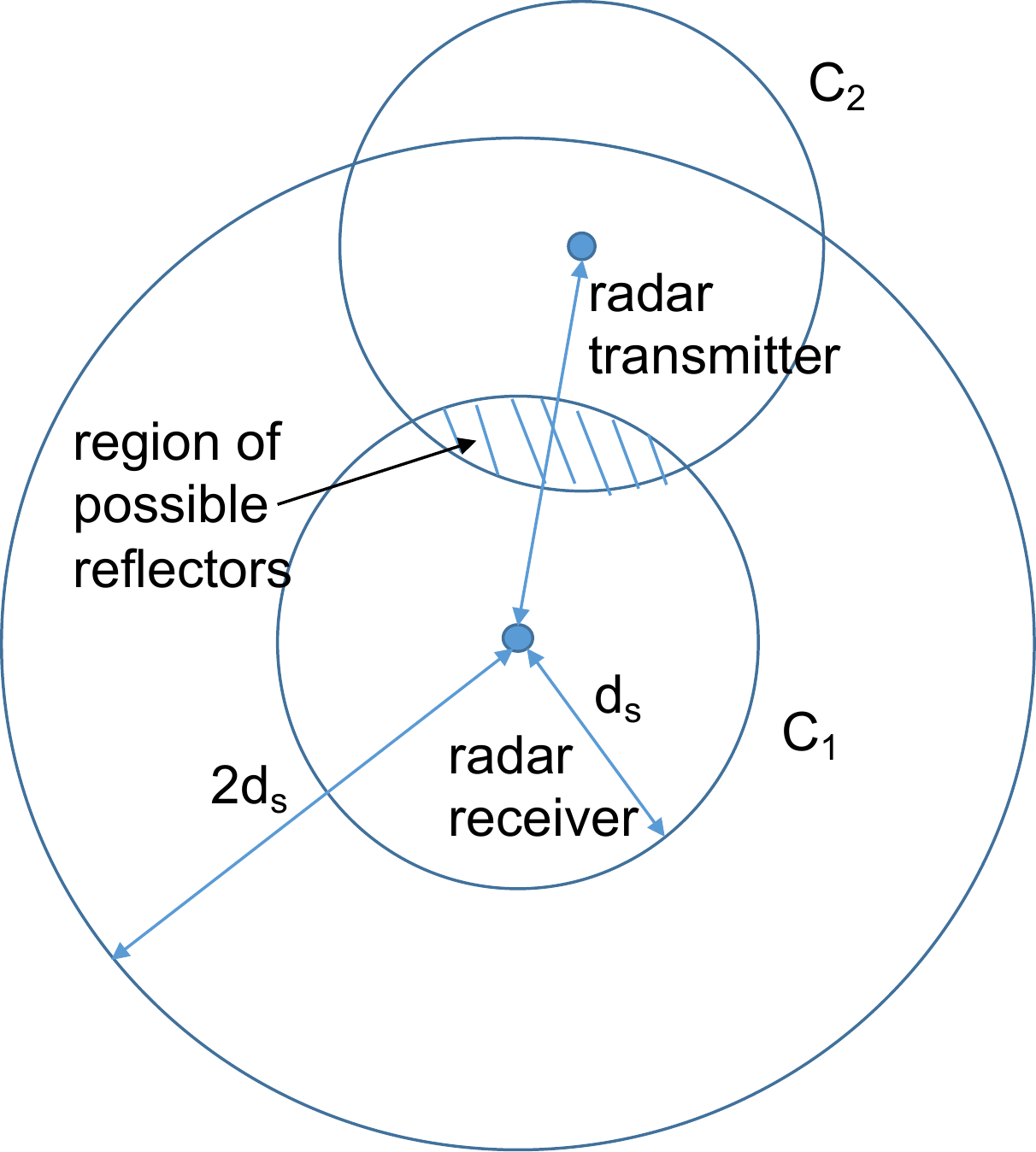}
\caption{Geometry of the transmit and receive radars}\label{fig:circles}
\label{pattern}
\end{figure}

Now, we need to evaluate $E[I_i]$ for an arbitrary radar transmitter $i$. We notice that all interfering radar transmitters are within the circle centered at the radar receiver with radius $2d_s$. The geometry is shown in Fig. \ref{fig:circles}: the distance between the radar transmitter and receiver is $r$ ($r<2d_s$), and all the reflecting objects that are illuminated by the transmitter and interfere the receiver are located within the shaded area. Therefore, the probability for transmitter $i$ to interfere the receiver is given by
\begin{eqnarray}
P(interfere)=\frac{area(C_1\cap C_2)}{area(C_2)}\left(\frac{2(T_p-T_{\min})}{T_p}\right).
\end{eqnarray}

Calculating the shaded area and integrating over the random position of the transmitter, we have
\begin{eqnarray}
E[I_i]&=&\left(\frac{T_p-T_{\min}}{T_p}\right)\nonumber\\
&\times&\int_0^{2d_s} \frac{2d_s^2\cos^{-1}\left(\frac{r}{2d_s}\right)-\frac{r}{2}\sqrt{d_s^2-\frac{r^2}{4}}}{(\pi d_s^2)^2}rdr.
\end{eqnarray}

By looking up tables of indefinite integral, we have
\begin{eqnarray}
\int_0^{2d_s} \frac{2d_s^2\cos^{-1}\left(\frac{r}{2d_s}\right)-\frac{r}{2}\sqrt{d_s^2-\frac{r^2}{4}}}{(\pi d_s^2)^2}rdr=\frac{1}{4\pi}.
\end{eqnarray}
This concludes the proof.
\end{proof}


\begin{thebibliography}{11}
\bibitem{Aydogdu2019}
C. Aydogdu, M. F. Keskin, N. Garcia, H. Wymersch and D. W. Bliss, ``RadChat: Spectrum sharing for automotive radar interference mitigation,'' {\em IEEE Trans. on Intelligent Transportation Systems}, vol.20, no.12, pp.1-14, Dec. 2019.


\bibitem{Aydogdu2020}
C. Aydogdu, G. K. Carvajal, O. Eirksson, H. Hellsten, H. Herbertsson, Mu. F. Keskin, E. Nilsson, M. Rdystrom, K. Vanas, and H. Wymeersch, ``Radar interference mitigation for automated driving,'' {\em IEEE Signal Processing Magazine}, vol. 37, no.4, pp.72--84, Sept. 2020.

\bibitem{Bechter2015}
J. Bechter and C. Waldschmidt, ``Automotive radar interference mitigation by reconstruction and cancellation of interference component,'' in {\em Proc. of IEEE International Conference on Microwave for Intelligence Mobility}, 2015.

\bibitem{Bo2011}
L. Bo, S. Yao, C. Zhou, ``Study of multistatic radar against velocity-deception jamming,'' in {\em International Conference on Electronics, Communications and Control (ICECC)}, 2011.


\bibitem{Brooker2007}
G. M. Brooker, ``Mutual interference of millimeter wave radar systems,'' {\em IEEE Trans. on Electromagnetism Compatibility}, vol. 49, no.1, pp.170--181, Feb 2007.

\bibitem{Chiang2007}
M. Chiang, S. H. Low, A. R. Calderbank, J. C. Doyle, ``Layering as optimization decomposition: A mathematical theory of network architectures,'' {\em Proceedings of the IEEE}, vol.95, no.1, Jan. 2007.


\bibitem{Chu2020}
P. Chu, J. A. Zhang, X. Wang, Z. Fei, G. Fang and D. Wang, ``Interference characterization and power optimization for automotive radar with directional antenna,'' {\em IEEE Trans. on Vehicular Technology}, vol. 69, no. 4, pp.3703--3716, April 2020.

\bibitem{Goodman2000}
D. Goodman and N. Mandayam, ``Power control for wireless data,'' {\em IEEE Personal Communications}, vol. 7, no.2, pp.48-54, 2000. 

\bibitem{Goppelt2011}
M. Goppelt, H. L. Blocher and W. Menzel, ``Analytical investigation of mutual interference between automotive FMCW radar sensors,'' in {\em Proc. of the 6th German Microwave Conference}, 2011.

\bibitem{Greco2008}
M. Greco, F. Gini, A. Farina, ``Radar detection and classification of jamming signals belonging to a cone class,'' {\em IEEE Trans. on Signal Processing}, vol.56, no.5, pp.1984--1993, April 2008.



\bibitem{Hourani2018}
A. Al-Hourani, R. J. Evans, S. Kandeepan, B. Moran and H. Eltom, ``Stochastic geometry methods for modeling automotive radar interference,'' {\em IEEE Trans. on Intelligent Transportation Systems}, vol.19, no.2, pp.333-344, Feb. 2018.

\bibitem{Kelly1998}
F. P. Kelly, A. K. Maulloo and D. K. H. Tan, ``Rate control for communication networks: Shadow prices, proportional fairness and stability,'' {\em Journal of the Operational Research Society}, vol.49, no.3, pp.237--252, March 1998

\bibitem{Kim2018}
G. Kim, J. Mun, J. Lee, ``A peer-to-peer interference analysis for automotive chirp sequence radars,'' {\em IEEE Trans. on Vehicular Technology}, vol. 67, no.9, Sept. 2019.

\bibitem{XiaojunLin2006}
X. Lin and N. B. Shroff, ``The Impact of Imperfect Scheduling on Cross-Layer Congestion Control in Wireless Networks,'' {\em IEEE/ACM Transactions on Networking}, vol. 14, no. 2, pp.302-315, April 2006


\bibitem{Luo2013}
T. N. Luo, C. E. Wu, Y. E. Chen, ``77-GHz CMOS automotive radar transceiver with anti-interference function,'' {\em IEEE Trans. on Circuits and Systems-1. Regular Papers}, vol.60, no.12, pp.3247-, 3255, Dec. 2013.


\bibitem{Maithripala2005}
D. H. A. Maithripala, S. Jayasuriya, ``Radar deception through phantom track generation,'' in {\em American Control Conference}, 2005.


\bibitem{Mani2019}
A. Mani and Z. Yang, ``Interference management in radar systems,'' Technical Report in Texas Instruments, 2019.

\bibitem{MOSARIM2012}
European Commission, {\em More Safty for All by Radar Interference Mitigation,} Technical Report, 2012.


\bibitem{Mun2018}
J. Mun, H. Kim and J. Lee, ``A deep learning approach for automotive radar interference mitigation,'' arxiv, 2019.

\bibitem{Munari2018}
A. Munari, L. Simic and M. Petrova, ``Stochastic geometry interference analysis of radar network performance,'' {\em IEEE Communications Letter}, vol. 22, no.11, Nov. 2018. 


\bibitem{NHTSA2018}
National Highway Traffic Safety Administration, {\em Radar Congestion Study}, Technical Report, 2018.

\bibitem{JunyiLi2013}
J. Li, X. Wu, and R. Laroia, {\em OFDMA Mobile Broadband Communications: A Systems Approach}, Cambridge University, 2013.

\bibitem{Palomar2006}
D. P. Palomar, and M. Chiang, ``A tutorial on decomposition methods for network utility,'' {\em IEEE Journal on Selected Areas in Communications}, vol. 24, no.8, Aug. 2006.


\bibitem{Proakis2018}
J. G. Proakis and M. Salehi, {\em Digital Communications}, 5th edition, McGraw-Hill Education, 2018.


\bibitem{Ram2020}
S. S. Ram, G. Singh and G. Ghatak, ``Estimating radar detection coverage probability of targets in a cluttered environment using stochastic geometry,'' in {\em Proc. of IEEE International Radar Conference}, 2020.

\bibitem{Schuerger2008}
J. Schuerger and D. Garmatyuk, ``Deception jamming modeling in radar sensor networks,'' in {\em IEEE Military Communication Conference}, 2008. 

\bibitem{Schuerger2009}
J. Schuerger and D. Garmatyuk, ``Performance of random OFDM radar signals in deception jamming scenarios,'' in {\em IEEE Radar Conference}, 2009. 


\bibitem{Skaria2019}
S. Skaria, A. Al-Hourani, R. J. Evans, K. Sithamparanathan and L. Parampalli, ``Interference mitigation in automative radars using pseudo-random cyclic orthogonal sequences,'' {\em Sensors}, vol. 19, 2019.

\bibitem{Smith2014}
J. R. Smith, {\em Introduction to Algebraic Geometry}, Five Dimensions Press, 2014.

\bibitem{Tian2018}
Z. Tian, B. Wen. L. Jin and Y. Tian, ``Radio frequency interference suppression algorithm in spatial domain for compact high-frequency radar,'' {\em IEEE Geoscience and Remote Sensing Letters}, vol.15, no.1, Jan. 2018.

\bibitem{Tse2005}
D. Tse and P. Viswanath, {\em Fundamentals of Wireless Communication}, Cambridge University Press, 2005.

\bibitem{Verdu1998}
S. Verdu, {\em Multiuser Detection}, Cambridge University Press, 1998.

\bibitem{Williams2012}
S. R. Williams, {\em Communication in Mechanism Design: A Differential Approach}, Cambridge University Press, 2012.

\bibitem{Yates1995}
R.D.  Yates,  ``A  framework  for  uplink  power  control  in  cellular  radio  systems,'' {\em IEEE Journal on Selected Areas in Communications}, vol. 13, no. 7, pp. 1341--1347, Sept. 1995. 


\end{thebibliography}
\end{document}